\documentclass[12pt]{article}
\usepackage{amssymb,amsmath,latexsym,amsthm,amsfonts}
\usepackage[english]{babel}
\usepackage[T1]{fontenc}
\usepackage{color,hyperref}
\usepackage{graphics}
\usepackage{graphicx}
\usepackage{multirow}
\usepackage[utf8]{inputenc}

\usepackage[colorinlistoftodos]{todonotes}

\newtheorem{theorem}{Theorem}[section]

\newtheorem{proposition}[theorem]{Proposition}

\newtheorem{assumption}{Assumption}

\def\1g{1\hskip -3pt \mbox{l}}

\small\normalsize

    \title{On the dependence structure of the trade/no trade sequence of illiquid assets}

\author{
{\sc
Hamdi Ra\"{\i}ssi\footnote{Instituto de Estad\'{\i}stica, PUCV,
Errazuriz 2734, Valpara\'{\i}so, CHILE; email: hamdi.raissi@pucv.cl. This author  acknowledges the ANID funding Fondecyt 1201898.}
\qquad \qquad
}
}

\begin{document}

    \maketitle


    \abstract{In this paper, 
    we propose to consider the dependence structure of the trade/no trade categorical sequence of individual illiquid stocks returns. The framework considered here is wide as constant and time-varying zero returns probability are allowed. The ability of our approach in highlighting illiquid stock's features is underlined for a variety of situations. More specifically, we show that long-run effects for the trade/no trade categorical sequence may be spuriously detected in presence of a non-constant zero returns probability. 
    Monte Carlo experiments, and the analysis of stocks taken from the Chilean financial market, illustrate the usefulness of the tools developed in the paper.}

\quad


\vspace*{.4cm}\textit{JEL Classification:} C13; C22; C58.

\section{Introduction}

In the time series econometrics literature, it is widely documented that neglected non-stationary behaviors can generate misleading assessments. A classical example is given by the spurious regression described in Phillips (1987). Since the seminal papers of Mikosch and St\u{a}ric\u{a} (2004), St\u{a}ric\u{a} and Granger (2005), and Engle and Rangel (2008), a huge amount of papers have explored the possibility that an unconditionally non-constant variance can display spurious long memory features. This led to contributions proposing new tools for a correct assessment of time series dynamics, see Phillips and Xu (2006), Cavaliere and Taylor (2008) or Patilea and Ra\"{\i}ssi (2014) among many others. However, at the best of our knowledge, there is no contribution in this spirit on the behavior of the daily zero returns probability of illiquid stocks. In this paper, we aim to investigate the dependence structure of the trade/no trade sequence, i.e the process $(a_t)$ defined such that $a_t=1$ if a daily price change is observed at $t$, and 0 if not. It is suggested that long range dependence effects, could possibly be explained by non-constant illiquidity levels in time, i.e., $P(a_t=1)$ is non-constant. 
As a consequence, tools for analyzing the dependence structure of the daily trade/no trade structure corrected from the non-stationary probability are developed.  
%

Illiquid stocks exhibiting a large amount of daily zero returns are commonly observed in all markets, see e.g. Lesmond (2005) for emerging markets. In order to motivate the above arguments, examples taken from the chilean stock market (Santiago Stock Exchange, SSE) are provided in Figure \ref{one}-\ref{four}.\footnote{The author is grateful to Andres Celedon for research assistance.}  The main stocks indexes of the SSE are the IPSA (the 30 most liquid stocks) and the IGPA (comprising the 30 stocks of the IPSA, plus other stocks according to some criteria). The stocks presented here are part of the IGPA, but not of the IPSA. The daily returns of the stocks are displayed together with their smoothed $P(a_t=1)$. In view of the different panels in Figure \ref{one}-\ref{three}, trends and abrupt breaks can be seen in the daily trade probability. Such behaviors can be backed up by a variety of facts. In Figure \ref{one}, the common increasing trends could be explained by a general development of the SSE. Indeed, during the period 2000-2008 the Chilean GDP growth was relatively elevated, which made the SSE reached some kind of maturity.
On the other hand, abrupt shifts are often the product of particular events in the company history. For instance, the quick increase of the daily price change probability for the Security bank stock, may be explained by the merger by absorption of the Dresdner Bank Lateinamerika in September 2004. In addition, the Security Bank issued more than 32.8 millions new stocks after the capital increase announced during the extraordinary shareholders meeting by December 29th, 2004. Conversely, the quick decrease of the daily trade probability for the Provida stock seems to be a consequence of the take-over bid of Metlife on Provida during September 2013. As a consequence, we can conclude that a structural break in the zero returns probability, has occurred during the above mentioned period, with some degree of confidence. In Figure \ref{three}, long-run decreasing behaviors can be observed. 
All these observations suggest to allow for a time-varying $P(a_t=1)$ when analyzing the dependence structure of $(a_t)$.
Note that investigating the illiquidity sequence dynamics can be useful to understand various facts on the underlying asset. The methodology developed here may be used to identify the relevant past of $(a_t)$ to consider for modelling financial time series. In a stationary framework, reference can be made to Moysiadis and Fokianos (2014) for predicting $(a_t)$, or the GARCH models with covariates considered in Francq and Thieu (2019).

The article is structured as follows. In Section \ref{diag-tools}, the general framework of the study is presented. We investigate the analysis of the dependence structure of the $(a_t)$ when $P(a_t=1)$ is constant in Section \ref{cst-case}. The consequences of neglecting a time-varying $P(a_t=1)$ in the dependence structure of the trade/no trade sequence are highlighted in Section \ref{tv-case}. Then, adaptive tools that adequately take into account a non-constant $P(a_t=1)$ are developed. Monte Carlo experiments and real data analysis illustrate our theoretical findings in Section \ref{num-exp}.

\section{The theoretical framework and diagnostic tools}
\label{diag-tools}

Let us consider a one-period profit-and-loss random variable $r_t$. It is assumed that $r_1,\dots,r_n$ are observed, with $n$ the sample size. Recall that $a_t=0$ if $r_t=0$, and $a_t=1$ otherwise. We make the following assumption to describe abrupt or gradual changes in the stock's illiquidity degree.

\begin{assumption}[Non constant probability]\label{tv-prob} The time-varying probabilities $P(a_t=1)$ are given by $g(t/n)$ where $g(\cdot)$ is a non-constant deter\-ministic function, such that $0< g(\cdot)<1$ on the interval $(0,1]$, and satisfies a piecewise Lipschitz condition on  $(0,1]$.\footnote{For $u<0$, the function is set constant, that is $g(u)=\lim_{u\downarrow0}g(u)$. Throughout the paper, the piecewise Lipschitz condition means: there exists a positive integer $p$ and some mutually disjoint intervals $I_1,\dots,I_{p}$ with $I_1\cup\dots\cup I_{p}=(0,1]$ such that $g(u)=\sum_{l=1}^{p} g_l(u){\bf 1}_{\{u\in I_l\}},$ $u\in(0,1],$ where $g(\cdot)$ is a Lipschitz smooth function on $I_1,\dots,I_{p},$ respectively.} 
\end{assumption}

The rescaling device of Dahlhaus (1997) is often used to describe long run effects (see Cavaliere and Taylor (2007), Xu and Phillips (2008), Patilea and Ra\"{i}ssi (2013) and Wang, Zhao and Li (2019) among others). Note that the double subscript is avoided to simplify the notations.
The specification in Assumption \ref{tv-prob} is quite general, as it allows for patterns commonly observed in practice, such as trends or abrupt breaks. As we are interested in testing the independence of $(a_t)$, the framework given by Assumption \ref{tv-prob} is sufficient, with no need to consider (possibly stochastic) probabilities conditional to some past information. Indeed, as usual a test is built under some null hypothesis, i.e. independent $(a_t)$ in our case. The tools developed here may be mostly used in an identifying step to some modelling task developed in, e.g. Moysiadis and Fokianos (2014). Our approach is similar to the continuous time series analysis methodology. Indeed, it is usual to study the correlation structure of $(r_t)$ (respectively powers of $(r_t)$) in a first step. Then, if some correlations are found significant, a model is estimated for the (stochastic) conditional expectation (respectively the volatility) using for instance ARMA (respectively GARCH) models. \\

\subsection{The constant zero returns probability case}
\label{cst-case}

In this part, we assume that $(a_t)$ is strictly stationary, i.e. the particular case where the $g(\cdot)$ function is constant. Then, for testing the short run dependence in the $(a_t)$ sequence, the following hypotheses are considered

$$H_0:P(a_ta_{t-h}=1)=P(a_t=1)^2,\:\mbox{for all}\:h\in\{1,\dots,m\}$$
vs.
$$H_1:\exists\:h\in\{1,\dots,m\},\:\mbox{such that}\:P(a_ta_{t-h}=1)\neq P(a_t=1)^2,$$
taking a relatively small $m$. The hypothesis $H_1$ suggests the presence of a dependence structure for $(a_t)$. Let us introduce

$$\widehat{A}_m=(\hat{\gamma}_a(1)/\hat{\gamma}_a(0),\dots,\hat{\gamma}_a(m)/\hat{\gamma}_a(0))',\quad\mbox{where}\quad
\hat{\gamma}_a(h)=\frac{1}{n}\sum_{t=1+h}^{n}a_ta_{t-h}.$$
The statistic

$$\mathcal{S}_m^{(a)}=n\widehat{A}_m'\widehat{A}_m,$$
can be used for deciding $H_0\:\mbox{vs.}\:H_1$ with small $m$. On the other hand, for large $m$, the components of $\widehat{A}_m$ may be plotted to assess some persistency, or long-run effects, in the dependence structure of $(a_t)$. From the above, the tools for assessing some serial dependence are built under $H_0$, that is an iid $(a_t)$ process. The following proposition gives the asymptotic behavior of $\widehat{A}_m$ in the stationary framework. The convergence in distribution is denoted by $\stackrel{d}{\longrightarrow}$.

\begin{proposition}\label{iid-test-prop}
Suppose that $(a_t)$ is iid. Then,
as $n\to\infty$ we have,
\begin{equation}\label{first-statement1}
\sqrt{n}\widehat{A}_m\stackrel{d}{\longrightarrow}\mathcal{N}(0,I_m),\quad\mbox{and}\quad \mathcal{S}_m^{(a)}\stackrel{d}{\longrightarrow}\chi_m^2,
\end{equation}
where $I_m$ stands for the identity matrix of dimension $m$.
\end{proposition}
The iid assumption is made to detect any short-run dependence structure in the $(a_t)$ sequence using $\mathcal{S}_m^{(a)}$. On the other hand, persistency or long-run effects may be highlighted plotting $\widehat{A}_m$ together with the confidence intervals obtained from (\ref{first-statement1}). 
Note that such a plot is similar to an autocorrelation function (ACF) plot. However, as the $(a_t)$ sequence is categorical, and since we decide that $(a_t)$ is an independent sequence if all the components of $\widehat{A}_m$ are not significant, we prefer to use the term "dependence plots". For instance, the dependence plots of the stationary Lipigas and CLC stocks in Figure \ref{two-ACF}, suggest that the $(a_t)$ sequence is 1-dependent for these two stocks. The reader is referred to Brockwell and Davis (2006), Definition 6.4.3, for the $k$-dependent processes.
We end this section by mentioning the ability of the $\mathcal{S}_m^{(a)}$ statistic to detect any dependence between daily trade/no trade events. The almost sure convergence is denoted by $\stackrel{a.s.}{\longrightarrow}$.

\begin{proposition}\label{iid-test-consist}
Suppose that $(a_t)$ is strictly stationary ergodic, such that $H_1$ holds true. Then, we have $\widehat{A}_m\stackrel{a.s.}{\longrightarrow}C$, where $C$ is a vector of constants with at least a non-zero component.
\end{proposition}

The proof of Proposition \ref{iid-test-consist} is a direct consequence of the Ergodic Theorem (see Francq and Zako\"{\i}an (2019), Theorem A.2), and is therefore skipped. The $Q_m$ test consists in rejecting $H_0$, if $\mathcal{S}_m^{(a)}>\chi_{m,1-\alpha}^2$, where $\chi_{m,1-\alpha}^2$ is the $(1-\alpha)$th quantile of the $\chi_{m}^2$ distribution.


\subsection{The non-constant zero returns probability case}
\label{tv-case}

Let us first underline, that the statistics based on the $\hat{\gamma}_a(h)$'s, are not adequate for investigating the dependence structure of the $(a_t)$ sequence, when $P(a_t=1)$ is time-varying. Indeed, under Assumption \ref{tv-prob}, and assuming that $(a_t)$ is independent, it can be shown that

\begin{equation}\label{third-cv}
\hat{\gamma}_a(h)\stackrel{a.s.}{\longrightarrow}\int_0^1g^2(s)ds-\left(\int_0^1g(s)ds\right)^2,\quad h>0,
\end{equation}
and using similar arguments to that of the proof of Proposition \ref{tv-test-prop-power} below.
If we assume that $g(\cdot)$ is non-constant, then $\int_0^1g^2(s)ds>\left(\int_0^1g(s)ds\right)^2$ in general. As a consequence, the tools introduced in the above section can lead to a spurious detection of a dependence, even for large lags $h$. In view of the stocks considered in this paper, changes in the non-zero returns probabilities are commonly observed in practice for a variety of facts. For this reason, new tools taking into account the framework given by Assumption \ref{tv-prob} are proposed.

In this part, we wish to test

$$\widetilde{H}_0:g(t/n)g((t-h)/n)=P(a_ta_{t-h}=1),\:\mbox{for all}\:h\in\{1,\dots,m\},\:\mbox{and all}\:t/n\in(0,1],$$
vs.
$$\widetilde{H}_1:\exists\:h\in\{1,\dots,m\},\:\mbox{and}\:0<a<b<1\:\mbox{such that}\:g(t/n)g((t-h)/n)\neq P(a_ta_{t-h}=1),$$
for all $t/n\in[a,b]$, taking $m$ small. In a first step, we suppose that the $P(a_t=1)$'s are known. Later, feasible statistics will be proposed. In view of $\widetilde{H}_1$, a cumulative sums (CUSUM) statistic as

\begin{equation}\label{unfeasible}
\widetilde{A}_m=\left(\sup_{u\in(0,1]}|\tilde{\gamma}_a(1,u)|,\dots,\sup_{u\in(0,1]}|\tilde{\gamma}_a(m,u)|\right)',
\end{equation}
should be used, where

$$\tilde{\gamma}_a(h,u)=(n-h)^{-1}\sum_{t=1+h}^{[nu]}(a_t-P(a_t=1))(a_{t-h}-P(a_{t-h}=1)),$$
and $[\cdot]$ denotes the integer part of a real number. Accordingly to $\widetilde{H}_0$, suppose that $(a_t)$ is independent, and assume that Assumption \ref{tv-prob} holds true. Then we have
\begin{equation}\label{first-statement}
\sqrt{n}\widetilde{A}_m\stackrel{d}{\longrightarrow}\left(\sup_{u\in(0,1]}|G(u)|,\dots,\sup_{u\in(0,1]}|G(u)|\right)',
\end{equation}
as $n\to\infty$, where $G(u)=\int_{0}^{u}\sigma(s)dB(s)$, $B(\cdot)$ is a standard Brownian motion, and $\sigma(s)^2=g^2(s)\left(1-\right.$ $\left.g(s)\right)^2$. The proof of the above result is provided in the Appendix.
It can be seen that the asymptotic distribution of $\widetilde{A}_m$ is non-standard, and have to be approximated using for instance bootstrap methods. Nevertheless, as the supremum is taken, and noting that the time-varying $P(a_t=1)$ should be estimated nonparametrically to make feasible tools, it can be difficult to control the type I error for finite samples. In addition, noting that the variance of $(a_t-P(a_t=1))(a_{t-h}-P(a_{t-h}=1))$ is not constant, the maximum value is more likely to be attained in high variance periods. This could make difficult to detect alternatives for periods where the variance is low.

For all these reasons, tools built using the full sample are considered, although $\tilde{\gamma}_a(h,1)$ leads to compare $\int_{0}^{1}g_h(s)ds$ and $\int_{0}^{1}g^2(s)ds$, where $g_h(s)=\lim_{n\to\infty}P(a_{[sn]}$ $a_{[(s-h/n)n]})$. Indeed, this allows to avoid the drawbacks described above, at the cost of a loss of power in the particular case $\int_{0}^{1}g_h(s)ds=\int_{0}^{1}g^2(s)ds$, with $g_h(\cdot)\neq g^2(\cdot)$. Then, let us introduce

\begin{equation}\label{unfeasible}
\overline{A}_m=\left(\tilde{\gamma}_a(1,1)/\tilde{\gamma}_a(0,1),\dots,\tilde{\gamma}_a(m,1)/\tilde{\gamma}_a(0,1)\right)'.
\end{equation}
In order to test $\int_{0}^{1}g_h(s)ds=\int_{0}^{1}g^2(s)ds$, $1\leq h\leq m$, taking $m$ small, the test statistic

\begin{equation}\label{test-stat-1}
\overline{\mathcal{S}}_m^{(a)}=n\hat{\omega}^{-1}\overline{A}_m'\overline{A}_m,
\end{equation}
where
$$\hat{\omega}:=\frac{n^{-1}\sum_{t=2}^{n}(a_t-P(a_t=1))^2(a_{t-1}-P(a_{t-1}=1))^2}{\left[n^{-1}\sum_{t=1}^{n}(a_t-P(a_t=1))^2\right]^2},$$
can be used. The following propositions give the asymptotic behavior of $\overline{A}_m$.

\begin{proposition}\label{tv-test-prop}
Suppose that the sequence $(a_t)$ is independent, and fulfills Assumption \ref{tv-prob}. Then as $n\to\infty$
$$\sqrt{n}\:\overline{A}_m\stackrel{d}{\longrightarrow}\mathcal{N}(0,\Omega),\quad \overline{\mathcal{S}}_m^{(a)}\stackrel{d}{\longrightarrow}\chi_m^2,$$
where $\Omega=diag(\omega,\dots,\omega)$, and
$$\omega=\frac{\int_0^1g^2(s)(1-g^2(s))^2ds}{\left(\int_0^1g(s)(1-g(s))ds\right)^2}.$$
In addition, we have $\hat{\omega}\stackrel{a.s.}{\longrightarrow}\omega.$
\end{proposition}

\begin{proposition}\label{tv-test-prop-power}
Suppose that Assumption \ref{tv-prob} holds true.
Under $\widetilde{H}_1$ with $\int_{0}^{1}g_h(s)ds\neq\int_{0}^{1}g^2(s)ds$, then as $n\to\infty$, $\overline{A}_m\stackrel{a.s.}{\longrightarrow}\overline{C}$, where $\overline{C}$ is a vector of constants with at least a non-zero component.
\end{proposition}

In view of $\widetilde{H}_0$, the process $(a_t)$ is assumed independent but not necessarily identically distributed in Proposition \ref{tv-test-prop}. In addition, it can be seen that taking the full sample in $\overline{A}_m$, we only need to apply a classical Heteroscedasticity Consistent (HC) correction to handle a time-varying $P(a_t=1)$. Considering $\sqrt{n}\hat{\omega}^{-\frac{1}{2}}\overline{A}_m$, the usual $(-1,96;1,96)$ bounds can be used to assess the dependence horizon of $(a_t)$.
For a small $m$, and considering a fixed asymptotic level $\alpha$, the $\overline{Q}_m$ test rejects $\widetilde{H}_0$ if $\overline{\mathcal{S}}_m^{(a)}>\chi_{m,1-\alpha}^2$, where we recall that $\chi_{m,1-\alpha}^2$ is the $(1-\alpha)$th quantile of the $\chi_m^2$ distribution. Note that the critical values for the statistic based on $\overline{A}_m$ are standard, on the contrary to the cumulative sums in $\widetilde{A}_m$. If $P(a_t=1)$ is constant, then $\omega=1$, and we retrieve the result of Proposition \ref{iid-test-prop}. Alternatively, taking a large $m$, a plot of $\overline{A}_m$ can be examined to analyze persistency or long-run effects in $(a_t)$. Proposition \ref{tv-test-prop-power} shows the ability of our tools to detect some dependence for $(a_t)$, but provided that $\int_{0}^{1}g_h(s)ds\neq\int_{0}^{1}g^2(s)ds$.\\

We now consider a feasible statistic for analyzing the dynamics of the $(a_t)$ process. Let us introduce first the kernel estimator of the non-constant probability

\begin{equation*}\label{nonparam}
\widehat{P(a_t=1)}=\sum_{i=1}^nw_{ti}a_i,
\end{equation*}
    with $w_{ti}=\left(\sum_{j=1}^nK_{tj}\right)^{-1}K_{ti}$, and
    $$K_{ti}=\left\{
                  \begin{array}{c}
                    K((t-i)/nb)\quad \mbox{if}\quad t\neq i\\
                    0  \quad\mbox{if}\quad t=i,\\
                  \end{array}
                \right.$$
where $K(\cdot)$ is a kernel function on the real line, and $b$ is the bandwidth fulfilling the following conditions.

\begin{assumption}[Kernel and bandwidth]\label{k-b}
\begin{itemize}
\item[(a)] $K(\cdot)$ is a continuous kernel function defined on the real line with compact support, such that $0\leq \sup_z K(z)<R$ for some finite real number $R$, and $\int_{-\infty}^{\infty}K(z)dz=1$.
\item[(b)] As $n\to\infty$, $nb^4+\frac{1}{nb^2}\to0$.
\end{itemize}
\end{assumption}

Plugin the above kernel estimator in (\ref{unfeasible}), we get the adaptive dependence structure estimation for the $m$ first lags:

\begin{equation}\label{unfeasible}
\check{A}_m=\left(\check{\gamma}_a(1,1)/\check{\gamma}_a(0,1),\dots,\check{\gamma}_a(m,1)/\check{\gamma}_a(0,1)\right)',
\end{equation}
where

$$\check{\gamma}_a(h,1)=(n-h)^{-1}\sum_{t=1+h}^{n}(a_t-\widehat{P(a_t=1)})(a_{t-h}-\widehat{P(a_{t-h}=1)}).$$
The following proposition states the asymptotic equivalence between $\overline{A}_m$ and $\check{A}_m$.

\begin{proposition}\label{equiv-prop}
Suppose that $(a_t)$ is independent, and such that Assumption \ref{tv-prob} and \ref{k-b} hold true. Then, as $n\to\infty$ we have,
$\sqrt{n}\left(\overline{A}_m-\check{A}_m\right)=o_p(1).$
\end{proposition}

We are now in position to introduce the following feasible test statistic
\begin{equation}\label{test-stat-2}
\check{\mathcal{S}}_m^{(a)}=n\check{\omega}^{-1}\check{A}_m'\check{A}_m,
\end{equation}
where

$$\check{\omega}:=\frac{n^{-1}\sum_{t=2}^{n}(a_t-\widehat{P(a_t=1)})^2(a_{t-1}-\widehat{P(a_{t-1}=1)})^2}{\left[n^{-1}\sum_{t=1}^{n}(a_t-\widehat{P(a_t=1)})^2\right]^2}.$$
From Proposition \ref{tv-test-prop} and \ref{equiv-prop}, at the asymptotic level $\alpha$, $\widetilde{H}_0$ is rejected if $\check{\mathcal{S}}_m^{(a)}>\chi_{m,1-\alpha}^2$. The corresponding test is denoted by $\check{Q}_m$.

\section{Numerical experiments}
\label{num-exp}

In this section, we first study the finite sample behaviors of the different tools introduced in the paper by means of Monte Carlo experiments. Then the real data taken from the Santiago stock exchange will be considered.

\subsection{Monte Carlo experiments}
\label{MC}

\subsubsection{Empirical size}
\label{ES}

We simulated $N=5000$ independent trajectories of independent $(a_t)$ according to the following cases:

\begin{itemize}
  \item[1-] The probability is constant: $P(a_t=1)=0.6$.
  \item[2-] The probability  is time-varying: $P(a_t=1)=g(t/n)$ where $g(r)=0.4$ for $(0,0.4]$, $g(r)=2r-0.4$ for $(0.4,0.6]$, and $g(r)=0.8$ for $(0.6,1]$.
\end{itemize}
The sample sizes are $n=200, 400, 800$. The nominal asymptotic level of the tests is $\alpha=5\%$. Tables \ref{testcstprob} and \ref{testtvprob} display the outputs of the $Q_m$ and $\check{Q}_m$ tests. Tables \ref{depcstprob} and \ref{deptvprob} give the relative frequencies where the components of $\widehat{A}_m$ and $\check{A}_m$ are outside its 95\% confidence bounds.

From Tables \ref{testcstprob} and \ref{depcstprob}, it can be seen that the adaptive $\check{Q}_m$ and $\check{A}_m$ have similar results to those of the $Q_m$ and $\widehat{A}_m$ when $P(a_t=1)$ is constant and $(a_t)$ is independent. This can be explained by the fact that all the tools presented in the paper are all valid in the case 1. Nevertheless, from Tables \ref{testtvprob} and \ref{deptvprob}, we can notice that the $Q_m$ and $\widehat{A}_m$ do not display satisfactory outputs as the relative rejection frequencies tend to 100\% as the sample size is increased. In particular, for large $h$, $\widehat{A}_m$ can spuriously suggest the existence of long memory effects for the trade/no trade sequence. Clearly this is the consequence of the non-constant zero return probability. In contrast, we can see that the adaptive $\check{Q}_m$ and $\check{A}_m$ have a good control of the type I error in general.

\subsubsection{Empirical power}
\label{EP}

We simulated $N=5000$ of the following sequence inspired by Romano and Thomb (1996): $a_t=\dot{a}_t\dot{a}_{t-1}$, where $(\dot{a}_t)$ is an iid sequence such that $P(\dot{a}_t=1)=0.6$ is constant. Note that in view of the empirical size part, the comparison is fair as $P(a_t=1)$ is constant. For conciseness, we only display the outputs for the $\check{Q}_m$ and $Q_m$ tests in Table \ref{testcstpower}. The results for the $\widehat{A}_m$ and $\check{A}_m$ lead to similar results. When the sample is small, we can note some loss of power of the
$\check{Q}_m$ when compared to the $Q_m$ test. This can be explained by the non-parametric estimation of $P(a_t=1)$. However, the abilities of detecting a dependence structure are similar for sample sizes commonly encountered in practice.

\subsection{Real data analysis}
\label{real data}

The dependence structure of different stocks taken from the Santiago Stock market are investigated. Recall that the behaviors of $P(a_t=1)$ for these stocks are described in the Introduction (in particular see Figures \ref{one}-\ref{four}). The sample sizes and the empirical means of the $a_t$'s are given in Table \ref{p2pp}. The effects of smooth long run changes and abrupt breaks are illustrated considering the Conchatoro, Cencosud, Security, Provida, Cruzados, Blanco y Negro stocks, see Figure \ref{one-ACF}. The stationary $(a_t)$ case is studied using the Lipigas and CLC stocks, see Figure \ref{two-ACF}. 

From Figure \ref{one-ACF}, it can be seen that the $\widehat{A}_m$ lead to detect long-run dependence in the $(a_t)$ process. Nevertheless, it is likely that such a long-run dependence is spurious as different kinds of non-constant $P(a_t=1)$ can be observed for these stocks. Conversely, when the time-varying $P(a_t=1)$ is adequately taken into account by using the $\overline{A}_m$, the dependence structure seems only short-run. Let us now study the stationary Lipigas and CLC stocks. From Figure \ref{two-ACF}, it emerges that when $P(a_t=1)$ seems constant, the $\widehat{A}_m$ and the $\overline{A}_m$ lead to the same conclusion: the dependence structure is again short-run. Note that for all the stocks which outputs are displayed in Figures \ref{one-ACF} and \ref{two-ACF}, the $\overline{Q}_5$ rejects the independence hypothesis at the 5\% level whether the $P(a_t=1)$ seems constant or not (not displayed here). 

\section{Conclusion}
\label{concl}
Determining the relevant past of the daily price changes/no change categorical sequence $(a_t)$ may be of interest for financial times series modelling. As an example, note that the $(a_t)$ sequence shares the clustering properties of the volatility in many cases, see Figures \ref{one}-\ref{four}. Hence, our tools could help to decide if $(a_t)$ can be used as a covariate in some volatility models, see e.g. Francq and Thieu (2019) for the GARCH-X models. Also, past values of such a series are considered to specify the conditional probability of price changes in Moysiadis and Fokianos (2014). Nevertheless, spurious persistency or long-run effects assessments for the daily price changes/no change sequence may be avoided taking into account a potential non-stationary behavior in the data. It is found that the dependence structure of $(a_t)$ is short-run in general. 

\section*{References}
\begin{description}
\item[]{\sc Brockwell, P.J., and Davis, R.A.} (2006) \textit{Times Series: Theory and Methods}. 2nd edition, Springer, New York.
\item[]{\sc Cavaliere, G., and Taylor, A.M.R.} (2007) Time-transformed unit-root tests for models with non-stationary volatility. \textit{Journal of Time Series Analysis} 29, 300-330.
\item[]{\sc Cavaliere, G., and Taylor, A.M.R.} (2008) Bootstrap Unit Root Tests for Time Series with Nonstationary Volatility. \textit{Econometric Theory} 24, 43-71.
\item[]{\sc Dahlhaus, R.} (1997) Fitting time series models to nonstationary processes. \textit{Annals of Statistics} 25, 1-37.
\item[]{\sc Davidson, J.} (1994) \textit{Stochastic limit theory.} Oxford University Press. New York.
\item[]{\sc Engle, R.F., and Rangel, J.G.} (2008) The spline GARCH model for unconditional volatility and its global macroeconomic causes. \textit{Review of Financial Studies} 21, 1187-1222.
\item[]{\sc Francq, C. and Thieu, L.Q.} (2019) QML inference for volatility models with covariates. \textit{Econometric Theory} 35, 37-72.
\item[]{\sc Francq, C., and Zako\"{i}an, J-M.} (2019) \textit{GARCH models : structure, statistical inference, and financial applications}. Wiley.
\item[]{\sc Hansen, B.E.} (1992) Convergence to stochastic integrals for dependent heterogeneous processes. \textit{Econometric Theory} 8, 489-500.
\item[]{\sc Lesmond, D.A.} (2005) Liquidity of emerging markets. \textit{Journal of Financial Economics} 77, 411--452.
\item[]{\sc Mikosch, T., and St\u{a}ric\u{a}, C.} (2004) Nonstationarities in financial time series, the long-
    range dependence, and the IGARCH effects. \textit{Review of Economics and Statistics} 86, 378-390.
\item[]{\sc Moysiadis, T., and Fokianos, K.} (2014) On binary and categorical time series models with feedback. \textit{Journal of Multivariate Analysis} 131, 209-228.
\item[]{\sc Patilea, V., and Ra\"{i}ssi, H.} (2013) Corrected portmanteau tests for VAR models with time-varying variance. \textit{Journal of Multivariate Analysis} 116, 190-207.
\item[] {\sc Patilea, V., and Ra\"{i}ssi, H.} (2014) Testing second order dynamics for autoregressive processes in presence of time-varying variance. \textit{Journal of the American Statistical Association} 109, 1099-1111.
\item[]{\sc Phillips, P.C.B.} (1987) Time series regression with a unit root. \textit{Econometrica} 55, 277-301.
\item[]{\sc Phillips, P.C.B., and Xu, K.L.} (2006) Inference in autoregression under heteroskedasticity. \textit{Journal of Time Series Analysis} 27, 289-308.
\item[]{\sc Romano, J. P., and Thombs, L. A.} (1996) Inference for autocorrelations under weak assumptions. \textit{Journal of the American Statistical Association} 91, 590-600.
\item[]{\sc Sen, P. K., and Singer, J. M.} (1993) \textit{Large Sample Methods In Statistics}. Chapman \& Hall.
\item[]{\sc St\u{a}ric\u{a}, C., and Granger, C.} (2005) Nonstationarities in stock returns. \textit{Review of Economics and Statistics} 87, 503-522.
\item[]{\sc Wang, S., Zhao, Q., and Li, Y.} (2019) Testing for no-cointegration under time-varying variance. \textit{Economics Letters} 182, 45-49.
\item[]{\sc Xu, K.L., and Phillips, P.C.B.} (2008) Adaptive estimation of autoregressive models with time-varying
    variances. \textit{Journal of Econometrics} 142, 265-280.
\end{description}

\newpage

\section*{Proofs}

\begin{proof}[Proof of Proposition \ref{iid-test-prop}]
Firstly, note that we have $\bar{a}-P(a_t=1)=O_p(n^{-\frac{1}{2}})$, so that

\begin{equation*}\label{rome}
n^{\frac{1}{2}}\hat{\gamma}_a(h)=n^{\frac{1}{2}}\bar{\gamma}_a(h)+o_p(1),\:\mbox{for all}\:h\in\{1,\dots,m\},
\end{equation*}
where $\bar{\gamma}_a(h)=(n-h)^{-1}\sum_{t=1+h}^{n}(a_t-P(a_t=1))(a_{t-h}-P(a_t=1))$. Let us define $\overline{A}_m=(\bar{\gamma}_a(1),\dots,\bar{\gamma}_a(m))'$. From the Central Limit Theorem (CLT) for martingale difference sequences (see Theorem A.3 in Francq and Zako\"{\i}an (2019)), we have

$$n^{\frac{1}{2}}\overline{A}_m\stackrel{d}{\longrightarrow}\mathcal{N}(0,\overline{\Sigma}),$$
where $\overline{\Sigma}$ is a $m\times m$ dimensional diagonal matrix, with diagonal component $P(a_t=1)^2(1-P(a_t=1))^2$. On the other hand, it is easy to see that $\hat{\gamma}_a(0)\stackrel{a.s.}{\longrightarrow}P(a_t=1)(1-P(a_t=1))$. Hence, the result follows from the Slutsky Lemma.
\end{proof}


\begin{proof}[Proof of (\ref{first-statement})]
Let $\xi_{h,t}=(a_t-P(a_t=1))(a_{t-h}-P(a_{t-h}=1))$ and

\begin{equation*}
\widetilde{\Gamma}_m(u)=\left(\tilde{\gamma}_a(1,u),\dots,\tilde{\gamma}_a(m,u)\right)'.
\end{equation*}
The sequence $(\xi_{h,t})$ is a martingale difference, such that $V(\xi_{h,t})=g^2(t/n)(1-g(t/n))^2+O(n^{-1})$, from the Lipschitz condition with a finite number of breaks in Assumption \ref{tv-prob}. Then from
Theorem 2.1 of Hansen (1992), we obtain

$$\widetilde{\Gamma}_m(u)\stackrel{d}{\longrightarrow}\left(G(u),\dots,G(u)\right)'.$$
The desired result follows from the Continuous Mapping Theorem.
\end{proof}

\begin{proof}[Proof of Proposition \ref{tv-test-prop}]
Note that $E\{(a_t-P(a_t=1))(a_{t-h}-P(a_{t-h}=1))\}=0$. Then using the Central Limit Theorem for independent but heterogeneous sequences, see Davidson (1994), Theorem 23.6, we have

$$\sqrt{n}\left(\tilde{\gamma}_a(1,1),\dots,\tilde{\gamma}_a(m,1)\right)'\stackrel{d}{\longrightarrow}N(0,\varpi),$$
where $\varpi=\int_0^1g^2(s)(1-g^2(s))^2ds$ is obtained using some computations, and since $(a_t)$ is independent. On the other hand, from the Kolmogorov SLLN for independent but non-identically random variables (see, for instance,  Sen and Singer (1993), Theorem 2.3.10), we have

$$\tilde{\gamma}_a(h,1)\stackrel{a.s.}{\longrightarrow}\int_0^1g(s)(1-g(s))ds.$$
Hence, the first result of Proposition \ref{tv-test-prop} follows from the Slutsky Lemma. Now, for the convergence of $\hat{\omega}$, using again the Kolmogorov SLLN and some computations, we have

$$n^{-1}\sum_{t=2}^{n}(a_t-P(a_t=1))^2(a_{t-1}-P(a_{t-1}=1))^2\stackrel{a.s.}{\longrightarrow}\int_0^1g^2(s)(1-g^2(s))^2ds,$$

$$n^{-1}\sum_{t=1}^{n}\left(a_t-P(a_t=1)\right)^2\stackrel{a.s.}{\longrightarrow}\int_0^1g(s)(1-g(s))ds.$$
\end{proof}

\begin{proof}[Proof of Proposition \ref{tv-test-prop-power}]
From the Kolmogorov SLLN for independent but non identically random variables (see, Sen and Singer (1993), Theorem 2.3.10), and using some computations, we have

$$\tilde{\gamma}_a(h,1)\stackrel{a.s.}{\longrightarrow}\int_0^1g_h(s)ds-\left(\int_0^1g(s)ds\right)^2,$$
from the dominated convergence Theorem, and for any $h\in\{1,\dots,m\}$. Hence, under $\widetilde{H}_1$ the result follows.
\end{proof}

\begin{proof}[Proof of Proposition \ref{equiv-prop}]
In this proof similar arguments to that of the proof of Theorem 2 in Xu and Phillips (2008) are considered. As the break number is finite, we assume that the function $g(\cdot)$ is continuous, without a loss of generality. Let us introduce the short notations

$$\hat{p}_t=\sum_{i=1}^{n}w_{ti}a_i,\bar{p}_t=\sum_{i=1}^{n}w_{ti}p_i,\quad p_t=P(a_t=1),$$
and $z_i=a_i-p_i$. We can write

$$\hat{p}_t-\bar{p}_t=\frac{\frac{1}{nb}\sum_{i=1}^{n}K_{ti}z_i}{\frac{1}{nb}\sum_{j=1}^{n}K_{tj}}.$$
From $\widetilde{H}_0$, $(z_i)$ is an independent process, such that $E(z_i)=0$. In addition, from Lemma A(c) in Xu and Phillips (2008), we have $\left(1/nb\right)\sum_{i=1}^{n}K_{ti}\to1$. In view of the above arguments, deduce that

\begin{eqnarray*}
E\left(\frac{1}{nb}\sum_{i=1}^{n}K_{ti}z_i\right)^2
&=&\frac{1}{(nb)^2}\sum_{i=1}^{n}K_{ti}^2E(z_i^2)\\
&\leq&\left(\frac{1}{nb}\right)\left(sup_iK_{ti}\right)\left(\frac{1}{nb}\sum_{i=1}^{n}K_{ti}\right)\\
&=&O\left(\frac{1}{nb}\right),
\end{eqnarray*}
Thus, we write

\begin{equation}\label{hatbar}
\hat{p}_t-\bar{p}_t=O_p\left(\frac{1}{\sqrt{nb}}\right).
\end{equation}
On the other hand, we have

\begin{eqnarray*}
&&\frac{1}{nb}\sum_{i=1}^nK_{[nr]i}p_i\\&=&\frac{1}{nb}\sum_{i=1}^{n}K\left(\frac{[nr]-i}{nb}\right)g\left(\frac{i}{n}\right)
\\&=&\frac{1}{b}\left[\int_{1/n}^{2/n}K\left(\frac{[nr]-[ns]}{nb}\right)g\left(\frac{[ns]}{n}\right)ds+\dots\right.\\
&&+\left.\int_{1}^{(n+1)/n}K\left(\frac{[nr]-[ns]}{nb}\right)g\left(\frac{[ns]}{n}\right)ds\right]\\&&
=\frac{1}{b}\left(\int_{1/n}^{(n+1)/n}K\left(\frac{[nr]-ns}{nb}\right)g\left(s\right)ds\right)+O\left(\frac{1}{nb}\right)\\&&
\stackrel{z=(s-r)/b}{=}\int_{(\frac{1}{n}-r)/b}^{(1+\frac{1}{n}-r)/b}
K\left(\frac{[nr]-n(r+bz)}{nb}\right)g\left(r+bz\right)dz+O\left(\frac{1}{nb}\right)\\&&
=\int_{-\infty}^{\infty}
K\left(\frac{[nr]-nr}{nb}-z\right)g\left(r+bz\right)dz+O\left(\frac{1}{nb}\right),
\end{eqnarray*}
where the last equality is obtained for small enough $b$, and since a compact support is assumed for $K(\cdot)$ in Assumption \ref{k-b}(a). Using the Lipschitz condition in Assumption \ref{tv-prob}, deduce that

\begin{equation}\label{eq11}
\frac{1}{nb}\sum_{i=1}^nK_{[nr]i}p_i=g(r)+O(b)+O\left(\frac{1}{nb}\right).
\end{equation}

Now, writing

\begin{eqnarray*}
&&n^{-\frac{1}{2}}\sum_{t=1+h}^{n}(a_t-\hat{p}_t)(a_{t-h}-\hat{p}_{t-h})-
n^{-\frac{1}{2}}\sum_{t=1+h}^{n}(a_t-p_t)(a_{t-h}-p_{t-h})\\&=&
n^{-\frac{1}{2}}\sum_{t=1+h}^{n}(a_t-p_t)(p_{t-h}-\bar{p}_{t-h})
+n^{-\frac{1}{2}}\sum_{t=1+h}^{n}(a_t-p_t)(\bar{p}_{t-h}-\hat{p}_{t-h})\\&+&
n^{-\frac{1}{2}}\sum_{t=1+h}^{n}(p_t-\bar{p}_t)(a_{t-h}-p_{t-h})
+n^{-\frac{1}{2}}\sum_{t=1+h}^{n}(p_t-\bar{p}_t)(p_{t-h}-\bar{p}_{t-h})\\&+&
n^{-\frac{1}{2}}\sum_{t=1+h}^{n}(p_t-\bar{p}_t)(\bar{p}_{t-h}-\hat{p}_{t-h})
+n^{-\frac{1}{2}}\sum_{t=1+h}^{n}(\bar{p}_{t}-\hat{p}_{t})(a_{t-h}-p_{t-h})\\&+&
n^{-\frac{1}{2}}\sum_{t=1+h}^{n}(\bar{p}_{t}-\hat{p}_{t})(p_{t-h}-\bar{p}_{t-h})
+n^{-\frac{1}{2}}\sum_{t=1+h}^{n}(\bar{p}_{t}-\hat{p}_{t})(\bar{p}_{t-h}-\hat{p}_{t-h}),
\end{eqnarray*}
and using (\ref{hatbar}) and (\ref{eq11}), the desired result follows.
\end{proof}


\newpage

\section*{Tables and Figures}

\begin{table}[hh]\!\!\!\!\!\!\!\!\!\!
\begin{center}
\caption{\small{The sample sizes $n$ and $\bar{a}=n^{-1}\sum_{t=1}^{n}a_t$ for different stocks taken from the Chilean stock market.}}
\footnotesize{\begin{tabular}{|c|c|c|}
\cline{2-3}
 \multicolumn{1}{c|}{ } & n & $\bar{a}$   \\
\hline
\tiny{CONCHATORO} & 5188 & 0.83 \\ \hline
\tiny{CENCOSUD} & 5198 &  0.89 \\ \hline\hline
\tiny{SECURITY} & 5188 & 0.62 \\ \hline
\tiny{PROVIDA} & 5179 & 0.59  \\ \hline\hline
\tiny{CRUZADOS} & 2584 & 0.29 \\ \hline
\tiny{BN} & 1958 & 0.48 \\ \hline\hline
\tiny{LIPIGAS} & 896 & 0.57 \\ \hline
\tiny{CLC} & 1958 & 0.49 \\ \hline
\end{tabular}}
\label{p2pp}
\end{center}
\end{table}

\begin{table}[hh]\!\!\!\!\!\!\!\!\!\!
\begin{center}
\caption{\small{The relative rejections frequencies of the tests for the independence of the $(a_t)$ sequence for 5 lags. The zero returns probability is constant with independent $(a_t)$.}}
\footnotesize{\begin{tabular}{|c|c|c|c|}
\hline
 $n$ & 200 & 400 & 800 \\
\hline
$Q_m$ & 4.56 & 4.70 & 4.90\\ \hline
$\check{Q}_m$ & 4.88 & 4.76 & 4.94\\ \hline
\end{tabular}}
\label{testcstprob}
\end{center}
\end{table}

\begin{table}[hh]\!\!\!\!\!\!\!\!\!\!
\begin{center}
\caption{\small{The same as above but for a time-varying zero returns probability with independent $(a_t)$.}}
\footnotesize{\begin{tabular}{|c|c|c|c|}
\hline
 $n$ & 200 & 400 & 800 \\
\hline
$Q_m$ & 81.74 & 98.42 & 100.00\\ \hline
$\check{Q}_m$ & 4.88 & 5.56 & 5.62 \\ \hline
\end{tabular}}
\label{testtvprob}
\end{center}
\end{table}

\begin{table}[hh]\!\!\!\!\!\!\!\!\!\!
\begin{center}
\caption{\small{The relative rejections frequencies of the tests for the independence of the $(a_t)$ sequence for 5 lags. The zero returns probability is constant with dependent $(a_t)$.}}
\footnotesize{\begin{tabular}{|c|c|c|c|c|}
\hline
 $n$ & 100 & 200 & 400 & 800 \\
\hline
$Q_m$ & 85.82 & 99.92 & 100.00 & 100.00\\ \hline
$\check{Q}_m$ &75.54 & 99.56 & 100.00 & 100.00\\ \hline
\end{tabular}}
\label{testcstpower}
\end{center}
\end{table}

\begin{table}[hh]\!\!\!\!\!\!\!\!\!\!
\begin{center}
\caption{\small{The relative frequencies for the components of $A_m$ and $A_m$ being outside the 95\% asymptotic confidence bounds. The constant probability case.}}
\footnotesize{\begin{tabular}{|c|c|c|c|c|c|c||c|c|c|c|}
\cline{2-10}
 \multicolumn{1}{c|}{ }&$h$ & 1 & 2 & 3 & 4 & 5 & 20 & 40 & 60 \\
\hline
\multirow{2}{*}{$n=200$}&
$\widehat{A}_m$ & 4.38 & 4.28 & 4.86 & 4.52 & 4.28 &  4.02 & 2.64 & 1.78 \\
&$\check{A}_m$ & 4.64 & 4.24 & 5.08 & 4.64 & 4.48 & 4.16 & 2.90 & 1.88 \\ \hline
\multirow{2}{*}{$n=400$}&
$\widehat{A}_m$ & 4.84 & 4.98 & 4.68 & 5.06 & 4.70 & 4.66 & 3.66 & 3.30 \\
&$\check{A}_m$ & 5.02 & 5.18 & 4.76 & 5.00 & 4.82 & 4.62 & 3.80 & 3.44 \\ \hline
\multirow{2}{*}{$n=800$}&
$\widehat{A}_m$ & 5.06 & 4.62 & 5.18 & 4.88 & 4.80 & 4.70 & 4.32 & 4.08\\
&$\check{A}_m$ & 5.10 & 4.76 & 5.24 & 4.84 & 4.84 & 5.02 & 4.34 & 4.14 \\ \hline
\end{tabular}}
\label{depcstprob}
\end{center}
\end{table}

\begin{table}[hh]\!\!\!\!\!\!\!\!\!\!
\begin{center}
\caption{\small{The same as above, but for the time-varying probability case.}}
\footnotesize{\begin{tabular}{|c|c|c|c|c|c|c||c|c|c|c|}
\cline{2-10}
 \multicolumn{1}{c|}{ }&$h$ & 1 & 2 & 3 & 4 & 5 & 20 & 40 & 60\\
\hline
\multirow{2}{*}{$n=200$}&
$\widehat{A}_m$ & 51.00 & 51.12 & 49.76 & 50.42 & 49.66 & 34.80 & 9.98 & 1.16\\
&$\check{A}_m$ & 5.24 & 4.66 & 4.64 & 4.66 & 4.58 & 3.84 & 3.62 & 2.48\\ \hline
\multirow{2}{*}{$n=400$}&
$\widehat{A}_m$ & 79.02 & 78.80 & 78.00 & 79.52 & 78.16 & 73.28 & 61.78 & 42.16 \\
&$\check{A}_m$ & 5.48 & 5.32 & 5.08 & 5.42 & 5.24 & 4.60 & 4.06 & 4.02\\ \hline
\multirow{2}{*}{$n=800$}&
$\widehat{A}_m$ & 97.06 &97.34 &97.14 &97.04 &97.20 & 96.10 &95.28 &92.70\\
&$\check{A}_m$ &  5.44 &5.42 &5.40 &5.34 &4.82 & 5.14 &4.92 &4.88\\ \hline
\end{tabular}}
\label{deptvprob}
\end{center}
\end{table}

\begin{figure}[h]\!\!\!\!\!\!\!\!\!\!
\vspace*{4cm}

\protect \includegraphics{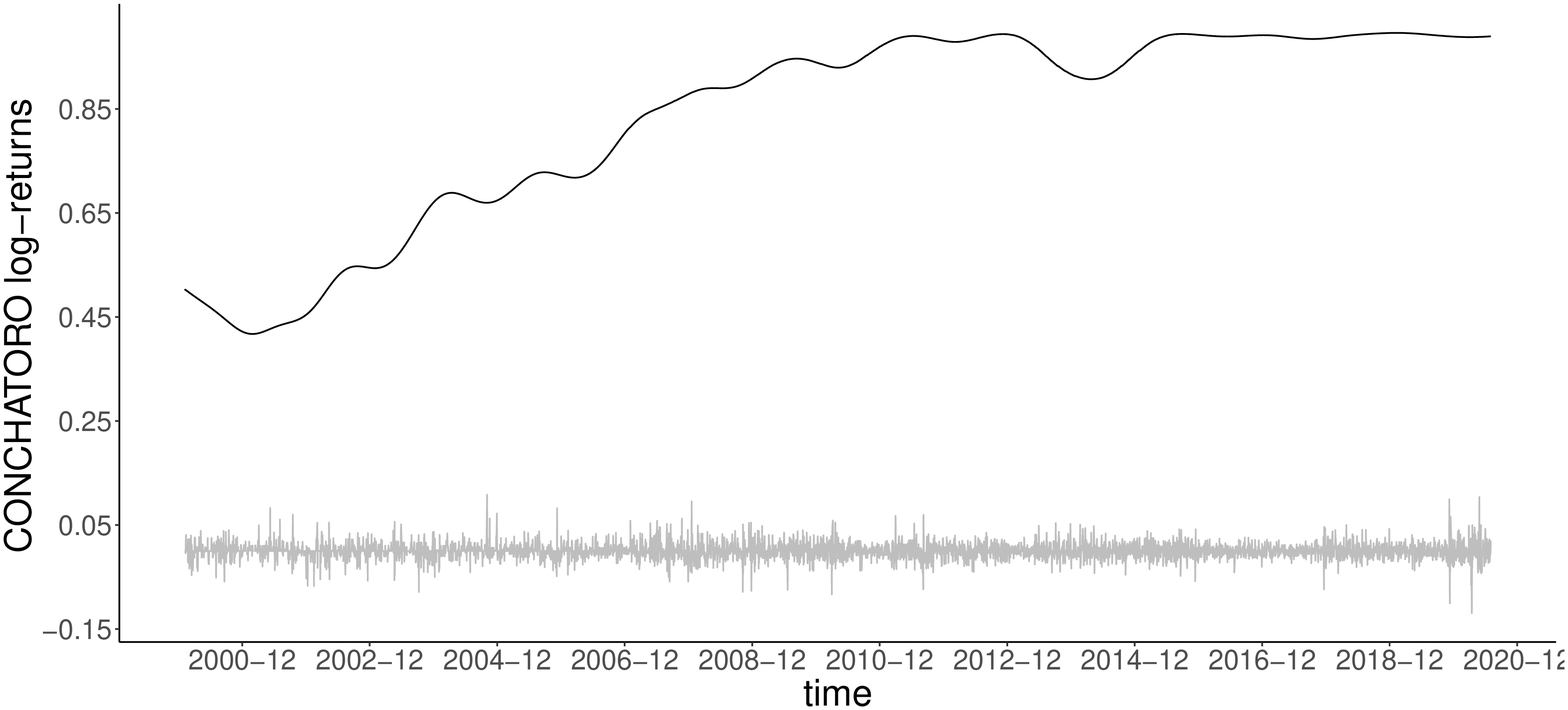}
\protect \includegraphics{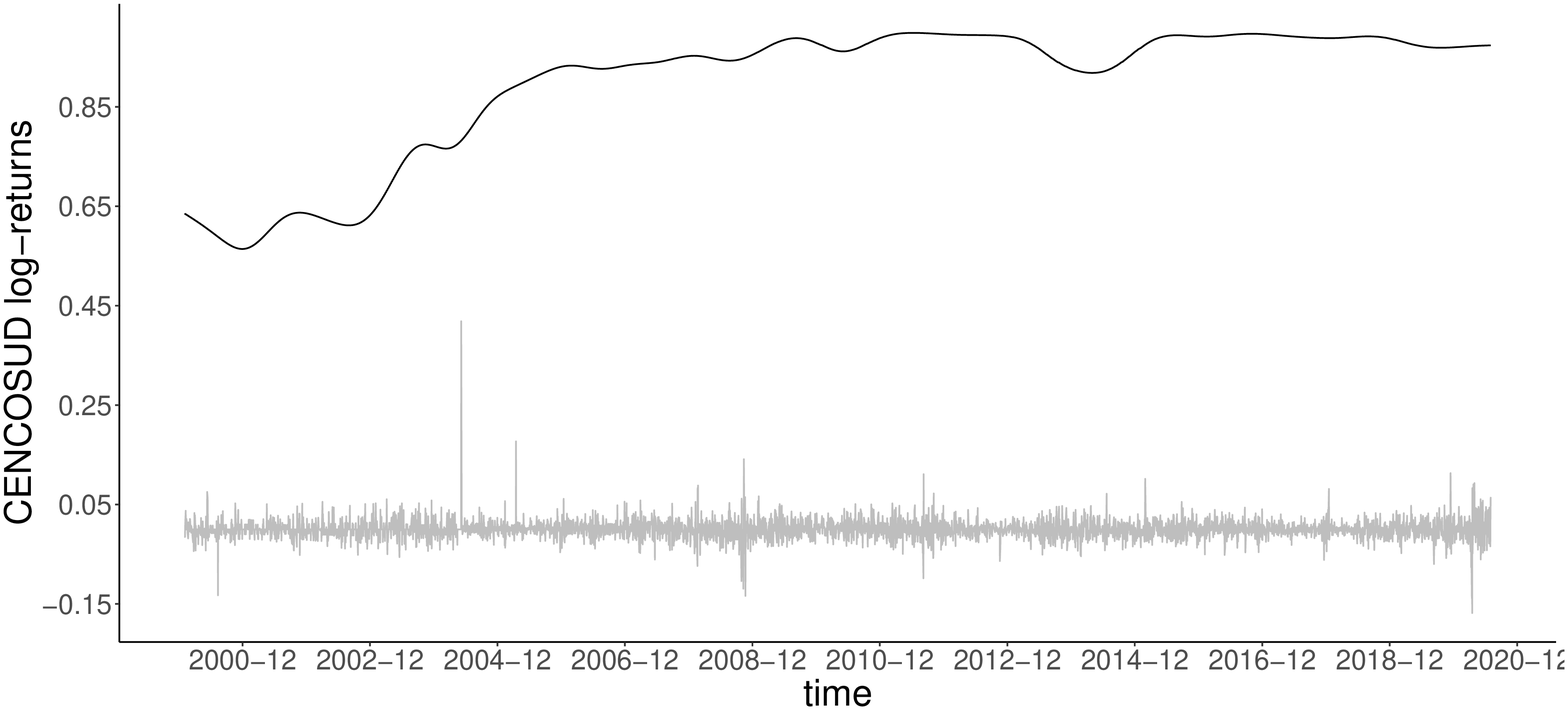}
\caption{\label{one}
{\footnotesize The log-returns of the Conchatoro and Cencosud stocks. A smooth increasing non-stationary behavior in the 2000's can be observed. The kernel smoothing of the $a_t$ values are displayed in full line. Data source: Yahoo Finance.}}
\vspace*{2cm}
\end{figure}

\begin{figure}[h]\!\!\!\!\!\!\!\!\!\!
\protect \includegraphics{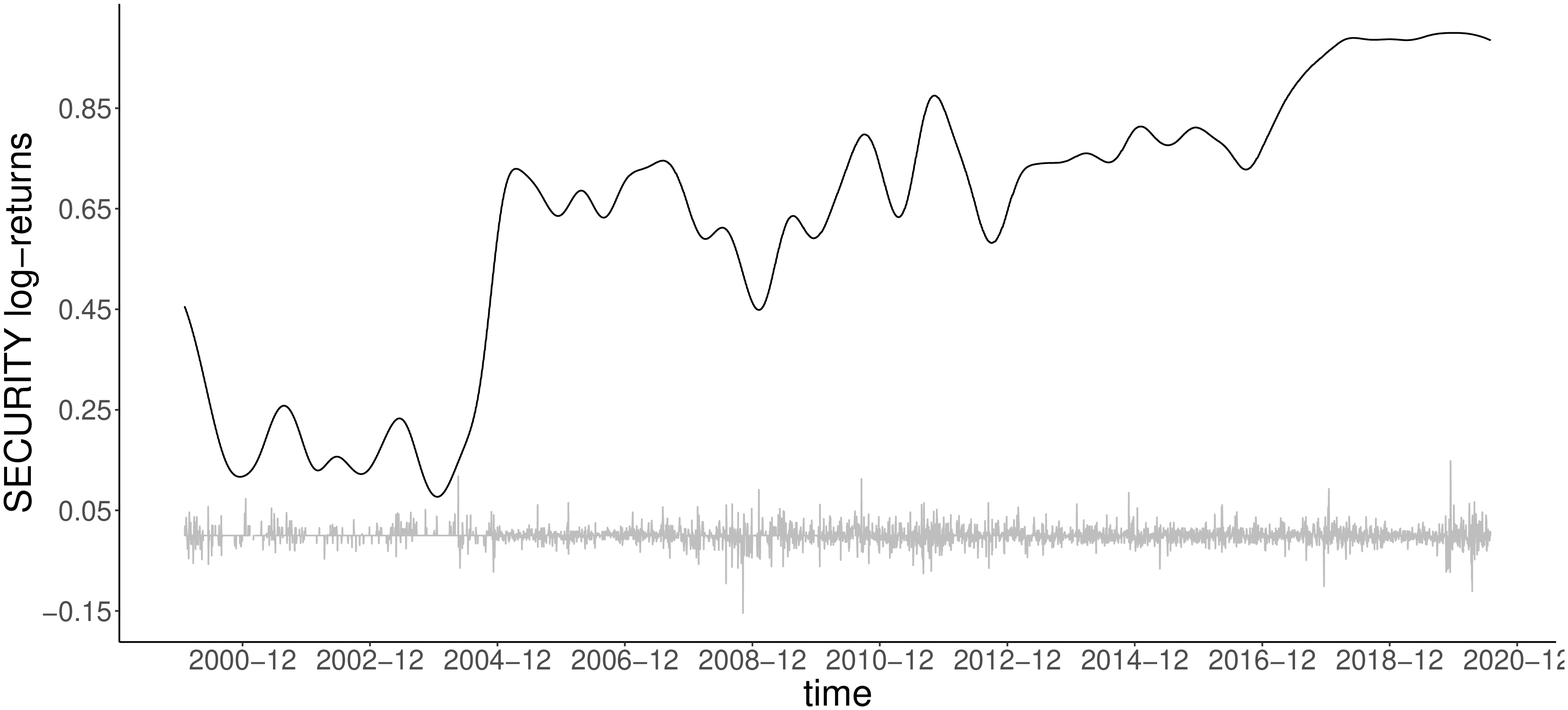}
\protect \includegraphics{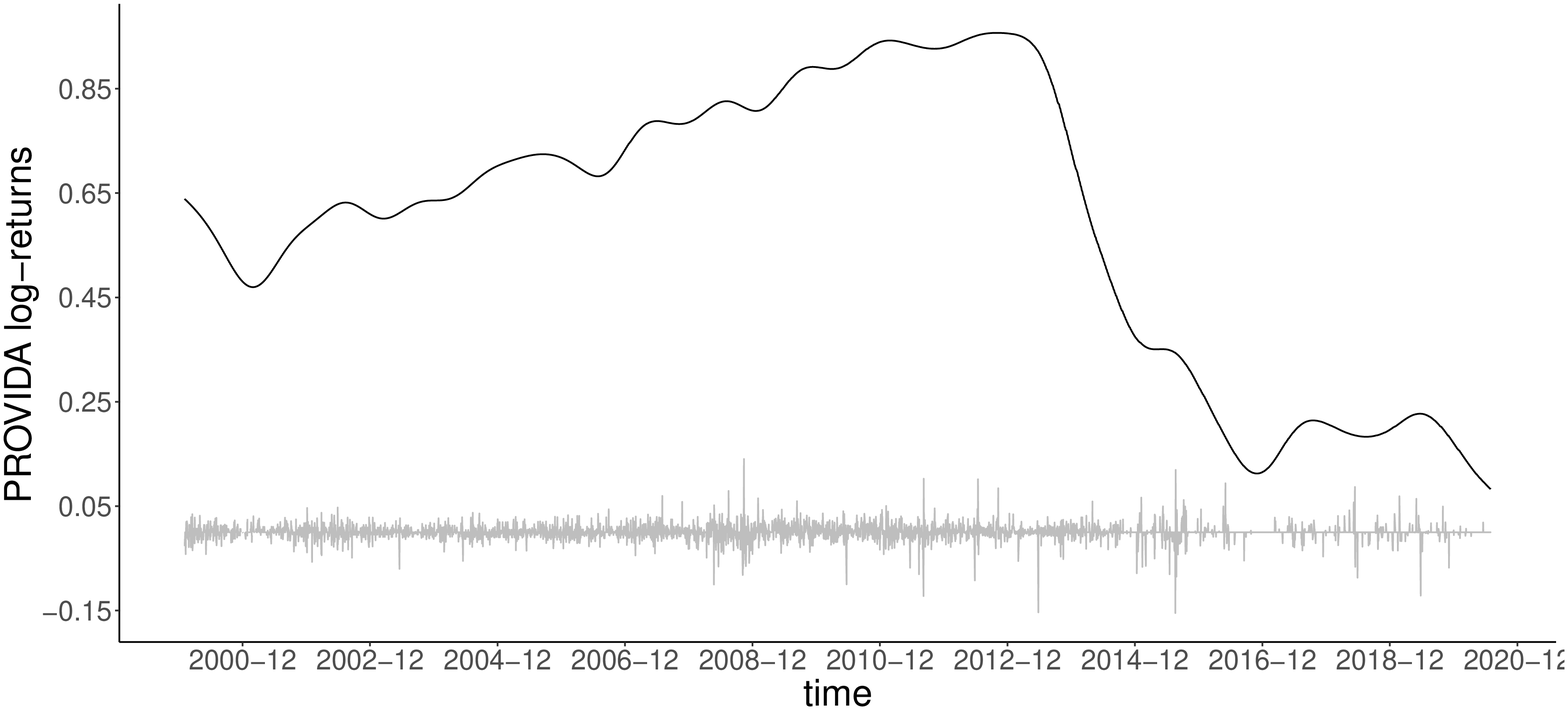}
\caption{\label{two}
{\footnotesize The same as for Figure \ref{one}, but for the Security and Provida stocks. The liquidity levels seem to display abrupt breaks due to specific events in the company's histories.}}
\vspace*{2.5cm}
\end{figure}

\clearpage

\vspace*{3cm}
\begin{figure}[h]\!\!\!\!\!\!\!\!\!\!
\protect \includegraphics{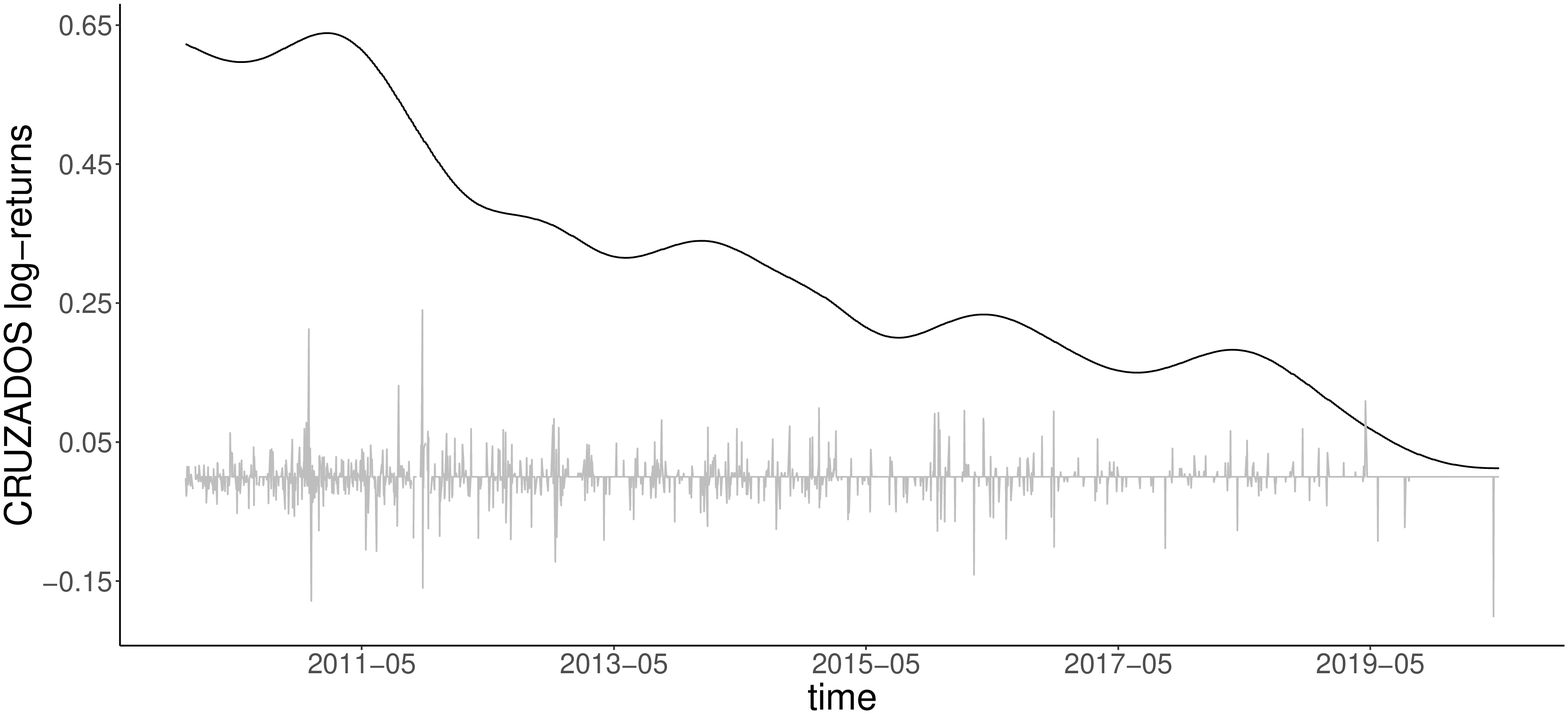}
\protect \includegraphics{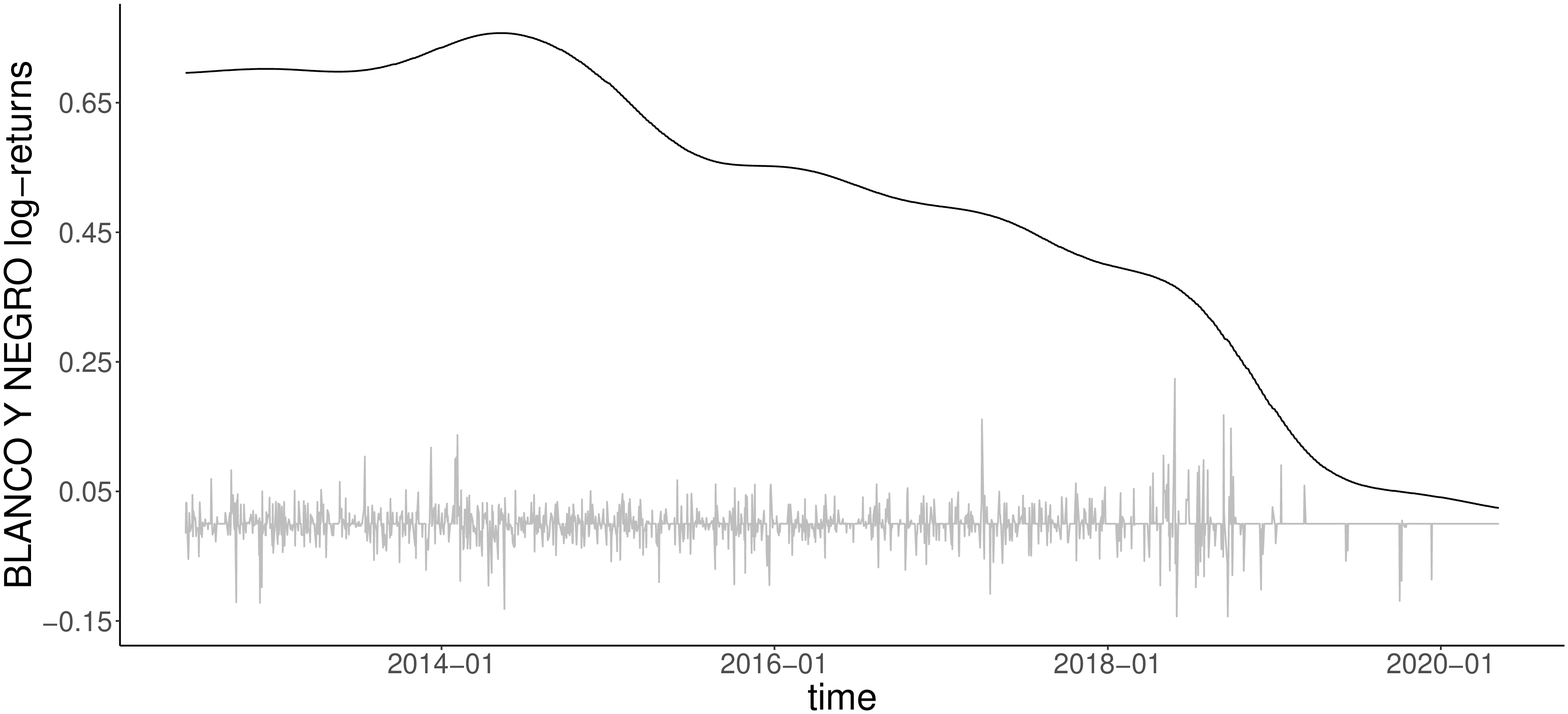}
\caption{\label{three}
{\footnotesize The same as for Figure \ref{one}, but for the Blanco y Negro and Cruzados stocks. The liquidity levels seem to have a smooth decreasing behavior. }}
\vspace*{2.5cm}
\end{figure}

\vspace*{3cm}
\begin{figure}[h]\!\!\!\!\!\!\!\!\!\!
\protect \includegraphics{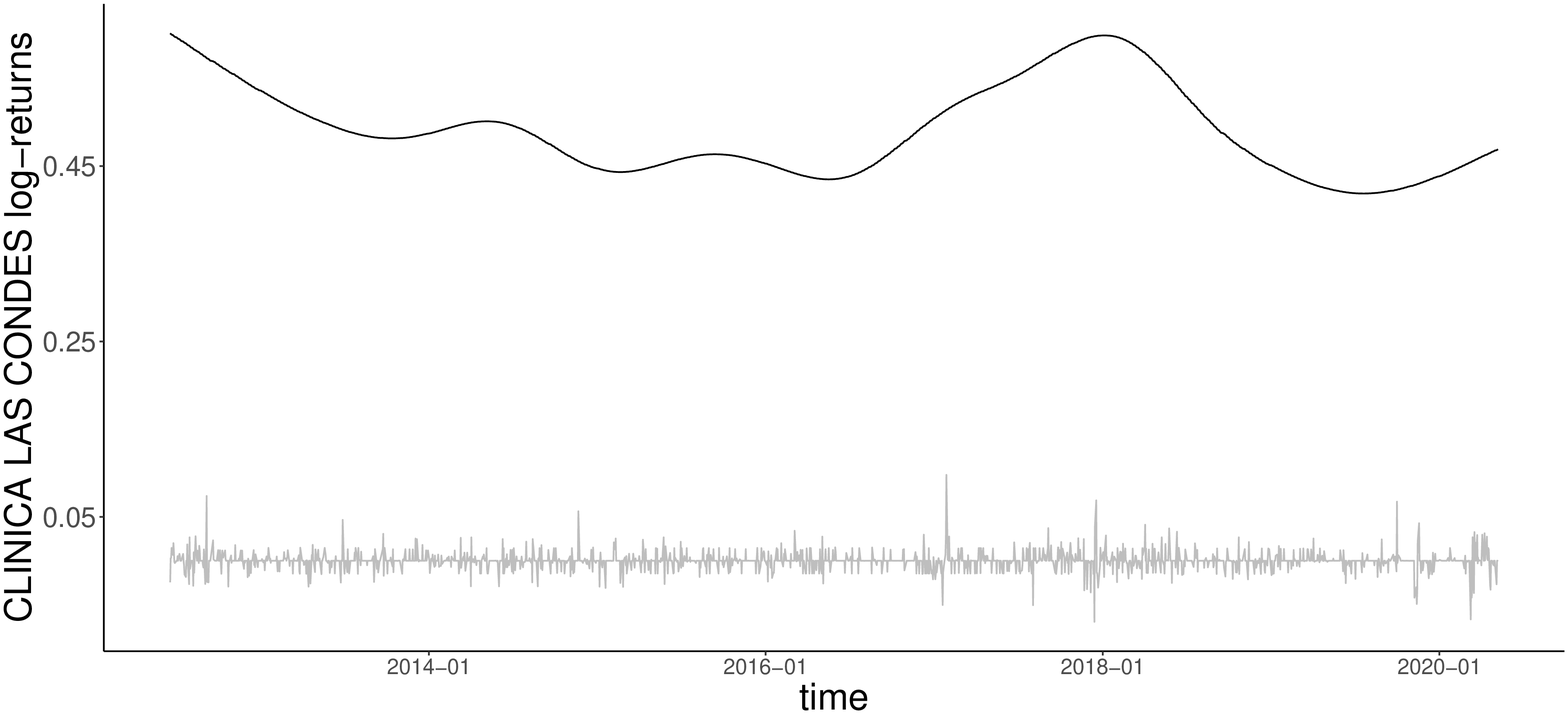}
\protect \includegraphics{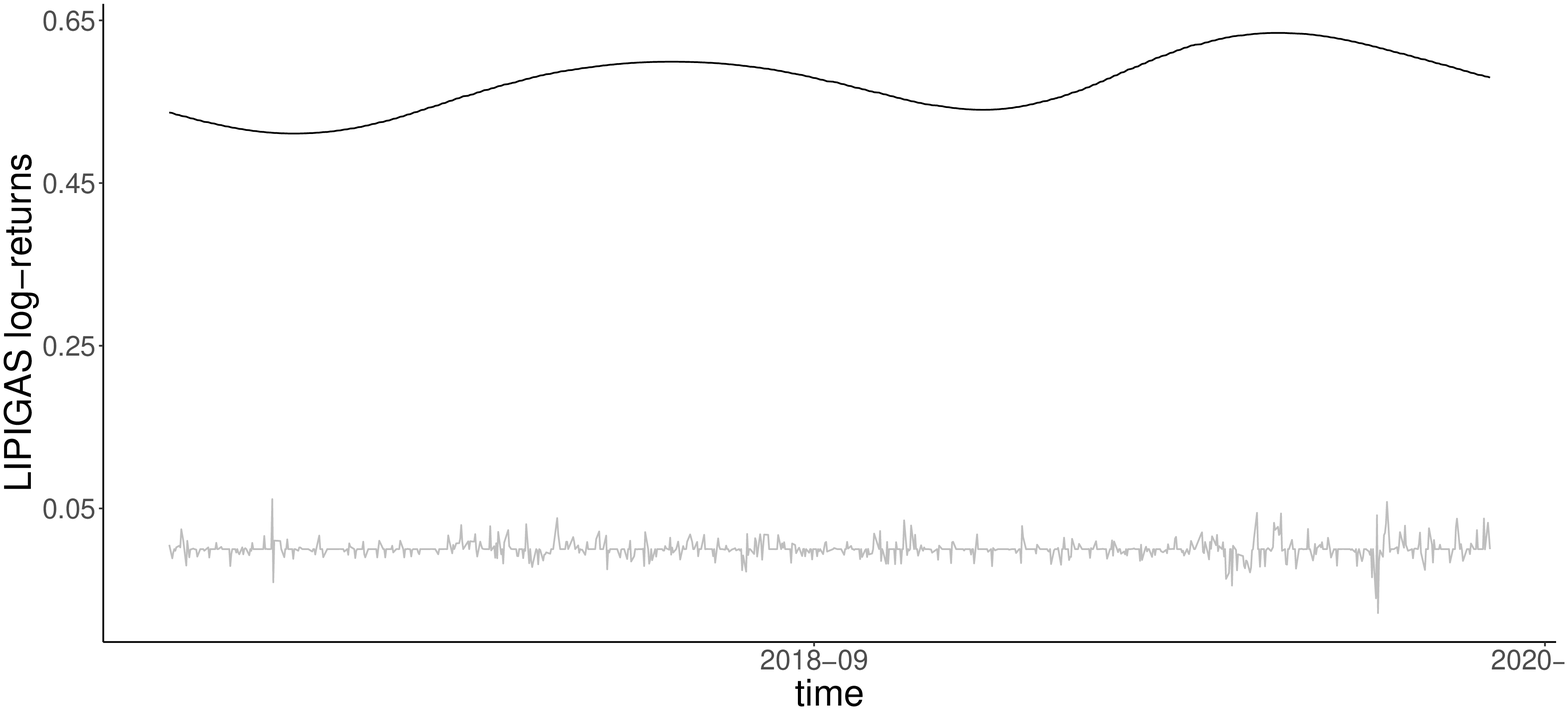}
\caption{\label{four}
{\footnotesize  The same as for Figure \ref{one}, but for the Lipigas and Clinica Las Condes stocks. The liquidity levels seem to be stationary.}}
\end{figure}

%
\begin{figure}[h]\!\!\!\!\!\!\!\!\!\!
\vspace*{18cm}

\protect \includegraphics{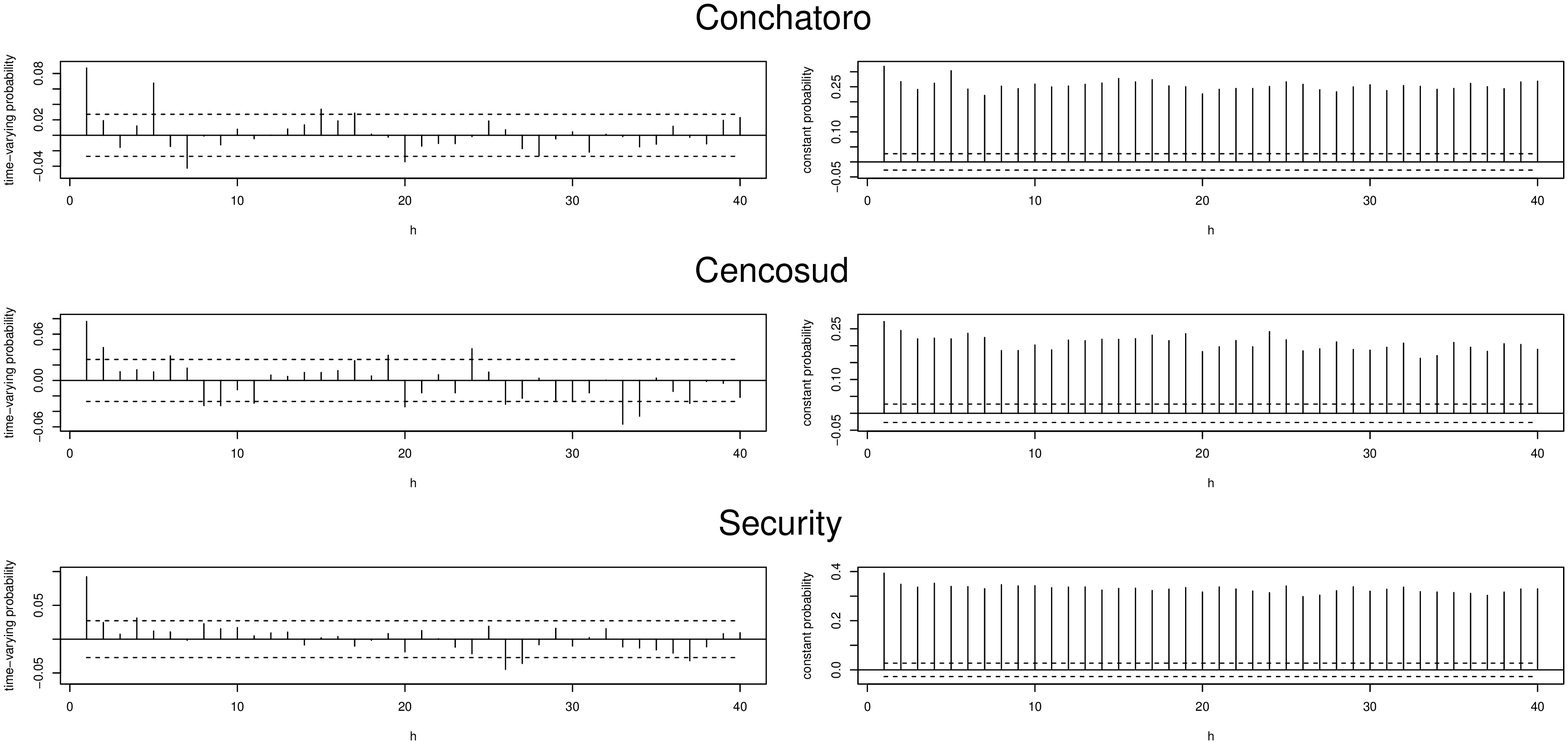}
\protect \includegraphics{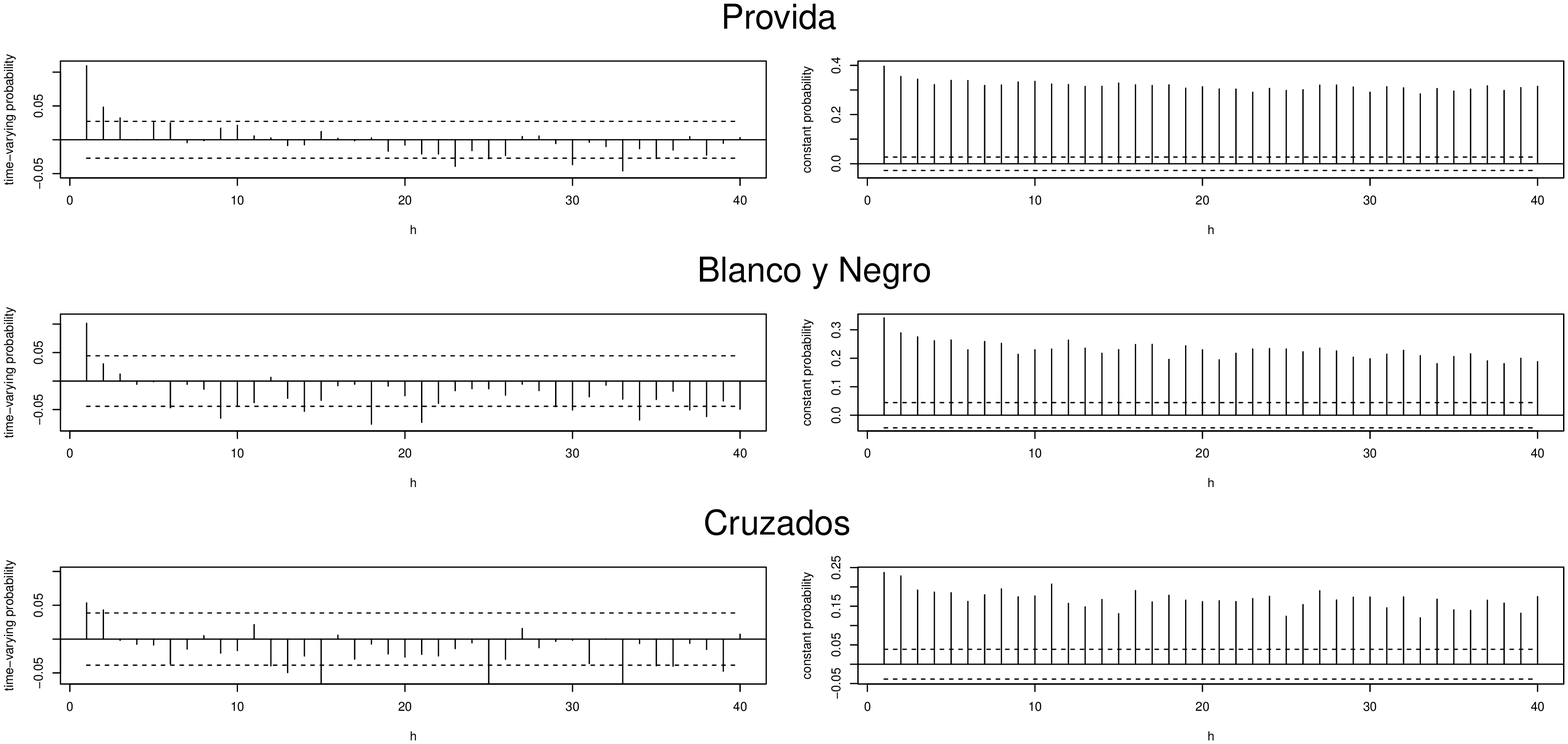}
\caption{\label{one-ACF}
{\footnotesize The dependence structure plots for stocks that seem to exhibit a non-stationary $(a_t)$.}}
\end{figure}

\begin{figure}[h]\!\!\!\!\!\!\!\!\!\!
\vspace*{9cm}

\protect \includegraphics{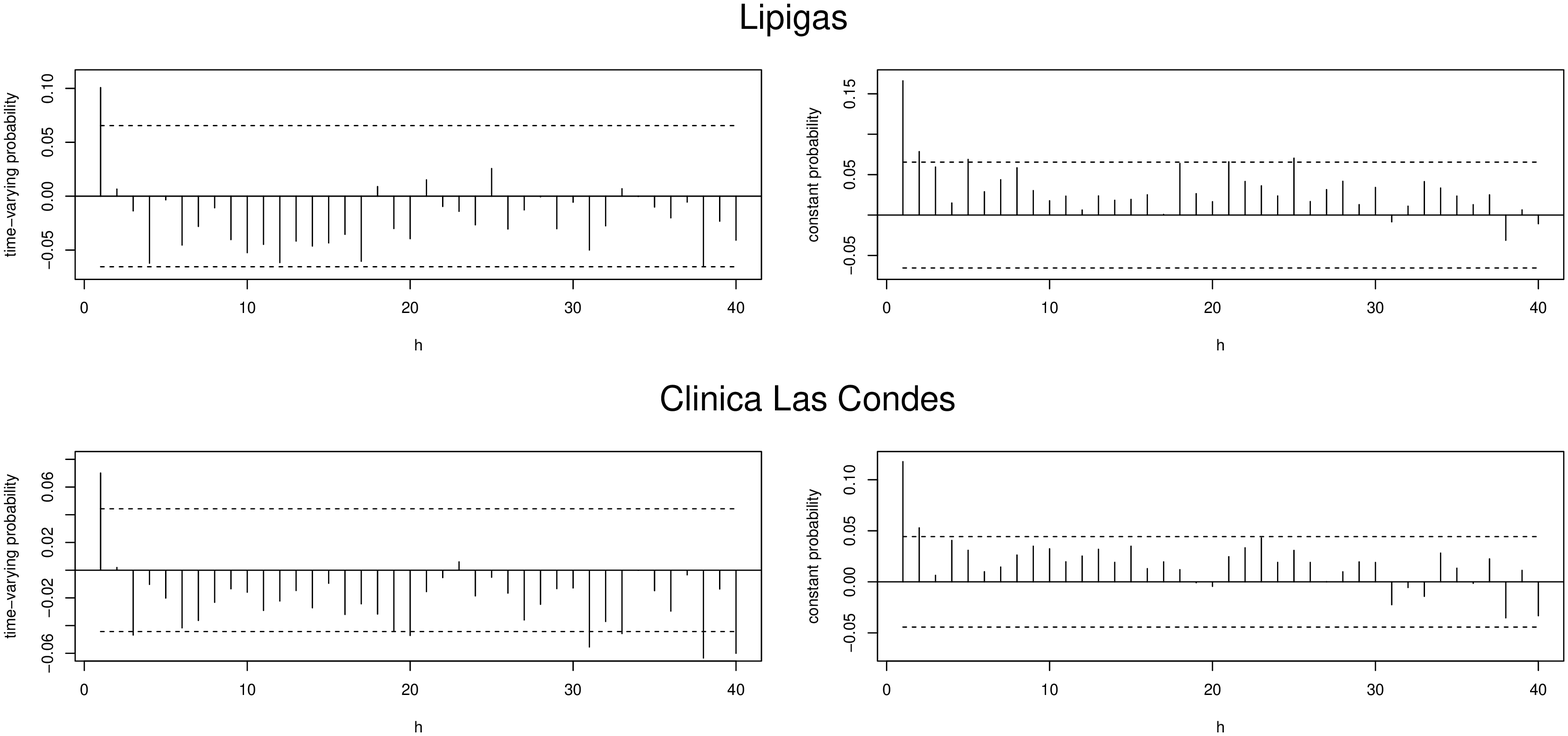}
\caption{\label{two-ACF}
{\footnotesize The dependence structure plots for stocks that seem to exhibit a stationary $(a_t)$.}}
\end{figure}

\end{document}